%% file: csConf.tex
\input{variables}

\documentclass[conference,comsoc]{IEEEtran}

\usepackage{multicol}
\usepackage{lipsum}
\usepackage{mathtools}
\usepackage{amsthm}
\usepackage{acronym}
\usepackage{stfloats}
\usepackage{amsfonts}
\usepackage{acronym}
\usepackage[noadjust]{cite}
\usepackage{multirow}
\usepackage{bm}
\usepackage{amsmath,amssymb}
\usepackage{graphicx}
\usepackage{epstopdf}
\usepackage[table,xcdraw]{xcolor}
\usepackage[caption=false,font=footnotesize]{subfig}
\usepackage[geometry]{ifsym}
\usepackage{array}
\usepackage[utf8]{inputenc}
\usepackage[T1]{fontenc}
\usepackage{pifont}
\usepackage{tabularray}
\usepackage{lipsum}
\def\BibTeX{{\rm B\kern-.05em{\sc i\kern-.025em b}\kern-.08em
		T\kern-.1667em\lower.7ex\hbox{E}\kern-.125emX}}

\input{otherDefinitions}

\input{acronyms}
\begin{document}

\title{
	{Majority Vote Computation With Complementary Sequences for Distributed UAV Guidance}\\
}

\author{
	\IEEEauthorblockN{Alphan \c{S}ahin and Xiaofeng~Wang} \IEEEauthorblockA{Electrical Engineering Department, University of South Carolina, Columbia, SC, USA\\
		Emails: asahin@mailbox.sc.edu, wangxi@cec.sc.edu}
} 

\maketitle

\begin{abstract}

This study introduces a novel non-coherent \ac{OAC} scheme aimed at achieving reliable \ac{MV} calculations in fading channels. The proposed approach relies on modulating the amplitude of the elements of \acp{CS} based on the sign of the parameters to be aggregated. Notably, our method eliminates the reliance on channel state information at the nodes, rendering it compatible with time-varying channels. To demonstrate the efficacy of our approach, we employ it in a scenario where an \ac{UAV} is guided by distributed sensors, relying on the MV computed using our proposed scheme. The experimental results confirm the superiority of our approach, as evidenced by a significant reduction in computation error rates in fading channels, particularly with longer sequence lengths. Meanwhile, we ensure that the peak-to-mean-envelope power ratio of the transmitted orthogonal frequency division multiplexing signals remains within or below 3~dB.
\end{abstract}
\section{Introduction}
\acresetall

Multi-user interference is often considered an undesired situation as it can degrade the performance of communications. In contrast, the same underlying phenomenon, i.e., the signal superposition property of wireless multiple-access channels, can be very useful in the computation of special mathematical functions, i.e., the family of nomographic functions such as mean, norm, polynomial function, maximum, and \ac{MV} \cite{goldenbaum2015nomographic}. The gain obtained with \ac{OAC} is that the resource usage can be reduced to a one-time cost, which otherwise scales with the number of nodes \cite{sahinSurvey2023}. Hence, OAC can benefit applications when a large number of nodes participate in computation over a bandwidth-limited wireless channel.

OAC was first analyzed by Bobak and Gastpar \cite{Nazer_2007} and applied to communication problems under interference channels and wireless sensor networks to improve spectrum utilization \cite{goldenbaum2013harnessing}.
 Recently, it has  gained momentum  for computation-oriented applications over wireless networks such as wireless federated learning  or wireless control systems. For example, the authors in \cite{chen2021distributed,Sahin_2022MVjournal,Guangxu_2020,Guangxu_2021}  implement \ac{FL} \cite{pmlr-v54-mcmahan17a} over a wireless network, where \ac{OAC} is used for aggregating gradients or model parameters of neural networks. In \cite{Park_2021}, difference equations are proposed to be computed with \ac{OAC}. In \cite{Cai_2018}, dynamic plants are considered along with \ac{OAC}. 
In \cite{Jihoon_2022Platooning}, \ac{OAC} is utilized to achieve mean consensus for  vehicle platooning. In \cite{Xiang_ojc2022}, multi-slot coherent \ac{OAC} is investigated for an \ac{UAV} network, where the \acp{UAV} compute the arithmetic mean of ground sensor readings with OAC. Similarly, in \cite{Min_uav2022}, \ac{UAV} trajectories are optimized based on the locations of the sensors.

A major challenge for OAC is computing functions in fading channels. To overcome fading channels, a large number of studies adopt pre-equalization techniques where the parameters are distorted with the reciprocals of the channel coefficients before the transmission so that the transmitted signals superpose coherently at the receiver \cite{Amiri_2020, Guangxu_2020,Guangxu_2021}. Although this approach has its merit, it  requires precise sample-level time synchronization as coherent signal superposition is sensitive to phase synchronization errors. Also, in \cite{Haejoon_2021}, it is shown that non-stationary  channel conditions in a \ac{UAV} network can severely degrade the coherent signal superposition. To relax this bottleneck, another approach is to use non-coherent techniques at the expense of sacrificing more resources. For instance, in \cite{Sahin_2022MVjournal} and \cite{safiMVlocOJCS}, orthogonal resources are allocated and energy accumulation on the resources are used for \ac{OAC}. In \cite{Goldenbaum_2013tcom}, random sequences are proposed to be utilized for OAC and the energy of the superposed sequences is calculated at the expense of interference components. A major challenge with  non-coherent \ac{OAC} is that the reliability cannot be improved directly by increasing the signal power due to the lack of pre-equalization. Hence, non-coherent OAC requires more investigations  for reducing computation errors.

In this study, we focus on computing \ac{MV} which can be made compatible with digital modulation due to its discrete nature,  and finds applications such as distributed training \cite{Sahin_2022MVjournal,Guangxu_2020,Guangxu_2021} and distributed localization \cite{safiMVlocOJCS}. We propose  a new non-coherent \ac{OAC} scheme based on \acp{CS} \cite{sahin_2020gm} to improve the robustness of computation against fading channels while limiting the dynamic range of transmitted signals to mitigate the distortion, like clipping, due to hardware non-linearity. Since the proposed approach does not rely on the availability of \ac{CSI} at the transmitter and receiver, it also provides robustness against time-varying channels. As opposed to earlier studies in \cite{Xiang_ojc2022, Min_uav2022, Haejoon_2021}, we demonstrate the applicability of the proposed method to a scenario where a \ac{UAV} is guided by distributed sensors based on MV computation.

\section{System Model}
\label{sec:system}
Consider a scenario where a \ac{UAV} flies from one point of interest $(\locationInitialEle[1],\locationInitialEle[2],\locationInitialEle[3])$ to another point of interest $(\locationTargetEle[1],\locationTargetEle[2],\locationTargetEle[3])$ in an indoor environment. Suppose that the \ac{UAV} cannot localize its location in the room. However, it can receive feedback from $\numberOfEdgeDevices$ distributed sensors deployed in the room about the velocity of the \ac{UAV} on $x$-, $y$-, and $z$-axis for every $\Trefresh$ seconds. Based on the feedback from the sensors, the \ac{UAV} updates its position at the $\indexRound$th round for the $x$-, $y$-, and $z$-axis, denoted by $\locationEle[\indexRound][1]$, $\locationEle[\indexRound][2]$, and $\locationEle[\indexRound][3]$, respectively, as
\begin{align}
	\locationEle[\indexRound+1][\indexCoordinate] &= \locationEle[\indexRound][\indexCoordinate] -\Trefresh\velocityVectorEle[\indexRound][\indexCoordinate],~\forall\indexCoordinate\in\{1,2,3\}, \label{eq:dynamics}\\
	\velocityVectorEle[\indexRound][\indexCoordinate]&= 
	\begin{cases}
		\max\{{\updateRate\feedbackEle[\indexRound][\indexCoordinate]},-\maximumVelocity\}& \feedbackEle[\indexRound][\indexCoordinate]<0\\
		\min\{{\updateRate\feedbackEle[\indexRound][\indexCoordinate]},\maximumVelocity\}& \feedbackEle[\indexRound][\indexCoordinate]\ge0
	\end{cases}
	\nonumber,
\end{align}
 for $\locationEle[0][\indexCoordinate]\triangleq\locationInitialEle[\indexCoordinate]$. In \eqref{eq:dynamics}, $\velocityVectorEle[\indexRound][\indexCoordinate]$ is the velocity at the $\indexRound$th round for the $\indexCoordinate$th coordinate,  $\maximumVelocity>0$ is the maximum velocity of the \ac{UAV}, $\updateRate>0$ is the update rate, $\feedbackEle[\indexRound][\indexCoordinate]$ is the velocity-update strategy given by
\begin{align}
	\feedbackEle[\indexRound][\indexCoordinate]=\begin{cases}
		\frac{1}{\numberOfEdgeDevices}\sum_{\indexED=1}^{\numberOfEdgeDevices}\locationEstimateEle[\indexRound][\indexCoordinate,\indexED]-\locationTargetEle[\indexCoordinate],&\text{Ideal (Continous)}\\
		\majorityVoteEle[\indexIteration],
		&\text{Ideal (MV)}\\
		\majorityVoteDetectedEle[\indexIteration],&\text{OAC (MV)}, \eqref{eq:decision}
	\end{cases}~,
	\nonumber
\end{align}
where 
$\locationEstimateEle[\indexRound][\indexCoordinate,\indexED]=\locationEle[\indexRound][\indexCoordinate]+\locationEstimationErrorEle[\indexRound][\indexCoordinate,\indexED]$ is the estimated position of the \ac{UAV} for the $\indexCoordinate$th coordinate at the $\indexED$th sensor, $\locationEstimationErrorEle[\indexRound][\indexCoordinate,\indexED]$ is a zero-mean Gaussian variable with the variance $\varianceSensor$, $\majorityVoteEle[\indexIteration]$ is the
ideal \ac{MV} function expressed as
$
	\majorityVoteEle[\indexIteration]=\signNormal[{\sum_{\indexED=1}^{\numberOfEdgeDevices}\voteVectorEDEle[\indexED][\indexIteration]}]
$, 
for $\voteVectorEDEle[\indexED][\indexIteration]=\signNormal[{\locationEstimateEle[\indexRound][\indexIteration,\indexED]-\locationTargetEle[\indexIteration]}]$, and $\majorityVoteDetectedEle[\indexIteration]$ denotes the  \ac{MV} obtained with the proposed OAC scheme in this work. 

\subsection{Complementary Sequences}
Let $\seqGa=(\eleGa[{\indexEleOfSeq}])_{i=0}^{\lengthGaGb-1}\triangleq(\eleGa[0],\eleGa[1],\dots, \eleGa[\lengthGaGb-1])$ be a sequence  of length $\lengthGaGb$ for $\eleGa[{\indexEleOfSeq}]\in\complexNumbers$ and $\eleGa[\lengthGaGb-1]\neq0$. We associate the sequence $\seqGa$ with the polynomial
$\polySeq[\seqGaP][\polyVariable] = \eleGa[\lengthGaGb-1]\polyVariable^{\lengthGaGb-1} + \eleGa[\lengthGaGb-2]\polyVariable^{\lengthGaGb-2}+ \dots + \eleGa[0]$
in indeterminate $\polyVariable$. The \ac{AACF} of the sequence $\seqGa$ given by
\begin{align}
	\apac[\seqGa][\lagForCorrelation]\triangleq
	\begin{cases}
		\sum_{\indexEleOfSeq=0}^{\lengthGaGb-\lagForCorrelation-1} \eleGa[\indexEleOfSeq]^*\eleGa[\indexEleOfSeq+\lagForCorrelation], & 0\le\lagForCorrelation\le\lengthGaGb-1\\
		\sum_{\indexEleOfSeq=0}^{\lengthGaGb+\lagForCorrelation-1} \eleGa[\indexEleOfSeq]\eleGa[\indexEleOfSeq-\lagForCorrelation]^*, & -\lengthGaGb+1\le\lagForCorrelation<0\\
		0,& \text{otherwise}
	\end{cases}~.
\end{align}
If $\apac[\seqGa][\lagForCorrelation]+\apac[\seqGb][\lagForCorrelation] = 0$ holds for $\lagForCorrelation\neq0$, 
the sequences $\seqGa$ and $\seqGb$ are referred to as \acp{CS}. It can be shown the \ac{PMEPR} of an \ac{OFDM} symbol constructed based on a \ac{CS} is less than or equal to 3~dB \cite{davis_1999}.

Let $\funcfForANF(\seqx)$ be a map from $\integers^\numberOfIterations_2=\{\seqx\triangleq(\monomial[1],\monomial[2],\dots, \monomial[\numberOfIterations])|\forall \monomial[\indexFirstOrderMonomial]\in\integers_2\}$ to $\realNumbers$ as $\funcfForANF:\integers^\numberOfIterations_2\rightarrow\realNumbers$, i.e., a pseudo-Boolean function.
 \acp{CS} can be obtained via pseudo-Boolean functions  as follows:
\begin{theorem}[{\cite{sahin_2020gm}}]
	\label{th:reduced}
	Let 
	$\seqPermutationCompShift=(\permutationMono[\indexIteration])_{\indexIteration=1}^{\numberOfIterations}$ be a permutation of $\{1,2,\dots,\numberOfIterations\}$. For any $\numberOfPointsForPSK,\numberOfIterations\in\integersPositive$, $\scaleEexp[\indexIteration],\arbitraryScaleE\in\realNumbers$, and $\angleexpAll[\indexIteration],\arbitraryPhaseK \in \integers_\numberOfPointsForPSK$ for $\indexIteration\in\{1,2,\mydots,\numberOfIterations\}$, let
	\begin{align}
&\funcfForFinalAmplitude(\seqx)
= \sum_{\indexIteration=1}^{\numberOfIterations}\scaleEexp[\indexIteration]\monomialAmp[{\permutationMono[{\indexIteration}]}]+\arbitraryScaleE\label{eq:realPartReduced}~,
\\		
		&\funcfForFinalPhase(\seqx)
		= {\frac{\numberOfPointsForPSK}{2}\sum_{\indexIteration=1}^{\numberOfIterations-1}\monomial[{\permutationMono[{\indexIteration}]}]\monomial[{\permutationMono[{\indexIteration+1}]}]}+\sum_{\indexIteration=1}^\numberOfIterations \angleexpAll[\indexIteration]\monomial[{\permutationMono[{\indexIteration}]}]+  \arbitraryPhaseK\label{eq:imagPartReduced}~,
	\end{align}
where  $
	\monomialAmp[{\permutationMono[{\indexIteration}]}]$ is $
		(\monomial[{\permutationMono[{\indexIteration}]}] +\monomial[{\permutationMono[{\indexIteration+1}]}])_2$ and 	$\monomial[{\permutationMono[{\numberOfIterations}]}]$ for $ \indexIteration<\numberOfIterations$ and $\indexIteration=\numberOfIterations$, respectively. 	Then, the sequence $\transmittedSeq[]=(\transmittedSeqEle[0],\mydots,\transmittedSeqEle[\lengthOfSequence-1])$, where  its associated polynomial  is  given by
	\begin{align}
		\polySeq[{\transmittedSeqP}][\polyVariable] &= 
		\sum_{\forall\seqx\in\integers^\numberOfIterations_2}\underbrace{
		\constante^{\funcfForFinalAmplitude(\seqx)}
		\constante^{\constantj\frac{2\pi}{\numberOfPointsForPSK}\funcfForFinalPhase(\seqx)}  }_{\transmittedSeqEle[\funcEnum(\seqx)]} 
		\polyVariable^{\funcEnum(\seqx)}~,\label{eq:encodedFOFDMonly}
	\end{align}
	is a \ac{CS} of length $\lengthOfSequence=2^\numberOfIterations$, where $\funcEnum(\seqx) \triangleq   \sum_{\indexFirstOrderMonomial=1}^{\numberOfIterations}\monomial[\indexFirstOrderMonomial]2^{\numberOfIterations-\indexFirstOrderMonomial}$, i.e.,  a decimal representation of the binary number constructed using all elements in the sequence $\seqx$. 
\end{theorem}
Theorem~\ref{th:reduced} shows that the functions that determine the amplitude and the phase of the elements of the CS $\transmittedSeq[]$ (i.e., $\funcfForFinalPhase(\seqx)$ and $\funcfForFinalPhase(\seqx)$) and \ac{RM} codes have similar structures. The function $\funcfForFinalPhase(\seqx)$ is in the form of the cosets of the first-order \ac{RM} code within the second-order \ac{RM} code  \cite{davis_1999}. Notice that the mapping between $\{(\monomialAmp[1],\mydots,\monomialAmp[\numberOfIterations])\}$ and $\{(\monomial[1],\mydots,\monomial[\numberOfIterations])\}$ is bijective and results in a Gray code when the elements of the set $\{(\monomial[1],\mydots,\monomial[\numberOfIterations])\}$ is ordered  lexicographically~\cite{sahin_2020gm}. Hence, the function $\funcfForFinalAmplitude(\seqx)$ is also similar to the first-order \ac{RM} code, except that the operations occur in $\realNumbers$.

\subsection{Signal model and wireless channel}
We assume that the sensors and the \ac{UAV}  are equipped with a single antenna.  Let  $\transmittedSeq[\indexED]=(\transmittedSeqEle[\indexED,0],\mydots,\transmittedSeqEle[\indexED,\lengthOfSequence-1])$ be a  \ac{CS} of length $\lengthOfSequence$ transmitted from the $\indexED$th sensor over an \ac{OFDM} symbol by mapping its elements to a set of contiguous subcarriers. Assuming that all sensors access the wireless channel simultaneously and the \ac{CP} duration is larger than the sum of maximum time-synchronization error and maximum-excess delay of the channel, we can express the polynomial representation of the received sequence $\receivedSeq[]=(\receivedSeqEle[0],\mydots,\receivedSeqEle[\lengthOfSequence-1])$ at the \ac{UAV} after the signal superposition as
\begin{align}
\polySeq[{\receivedSeqP}][\polyVariable] &=\sum_{\varMonomial=0}^{\lengthGaGb-1} \underbrace{\left(\sum_{\indexED=1}^{\numberOfEdgeDevices}
\channelAtSubcarrier[\indexED,\varMonomial] \sqrt{\transmitPower[\indexED]}
\transmittedSeqEle[\indexED,\varMonomial]+\noiseAtSubcarrier[\varMonomial]\right)}_{\receivedSeqEle[\varMonomial]}
\polyVariable^{\varMonomial}~,	
	\label{eq:symbolOnSubcarrier}
\end{align}
where  $\channelAtSubcarrier[\indexED,\varMonomial]\sim\complexGaussian[0][1]$ is the Rayleigh fading channel coefficient between the \ac{UAV} and the $\indexED$th sensor for the $\varMonomial$th element of the sequence unless otherwise stated, $\transmitPower[\indexED]$ is the average transmit power, and ${\noiseAtSubcarrier[\varMonomial]}\sim\complexGaussian[0][\noiseVariance]$ is the \ac{AWGN}. We assume that the {\em average} received signal powers of the sensors at the \ac{UAV} are aligned with a power control mechanism \cite{10.5555/3294673}. Hence, the relative positions of the sensors with respect to the UAV does not effect our analyses. Without loss of generality, we set $\transmitPower[\indexED]$, $\forall\indexED$, to $1$~Watt and calculate the \ac{SNR} of a sensor at the \ac{UAV} receiver as $\SNR=1/\noiseVariance$.

\subsection{Problem Statement}

Suppose that the fading coefficient $\channelAtSubcarrier[\indexED,\varMonomial]$  is not available at the $\indexED$th sensor or the \ac{UAV} due to the synchronization impairments, reciprocity calibration errors, or  mobility. Under this constraint, 
for $\indexIteration\in\{1,\mydots,\numberOfIterations\}$, the main objective of the \ac{UAV} is to compute the \ac{MV} $\majorityVoteEle[\indexIteration]$, $\forall\indexIteration$ by exploiting the signal superposition property of the \ac{MAC}, where $\voteVectorEDEle[\indexED][\indexIteration]\in\{-1,0,1\}$ is the $\indexED$th sensor's vote for the $\indexIteration$th \ac{MV} computation and $\majorityVoteEle[\indexIteration]\in\{-1,1\}$ is the $\indexIteration$th \ac{MV}. 
It is worth emphasizing that the $\indexED$th sensor does not participate in the \ac{MV} computation for $\voteVectorEDEle[\indexED][\indexIteration]=0$. 
Since we consider a single UAV in this work,  we set $\voteVectorEDEle[\indexED][\indexIteration]=0$ for $\indexIteration\in\{3,4,\mydots,\numberOfIterations\}$. Note that, in the literature, absentee votes are shown to be useful  for addressing data heterogeneity for distributed training scenarios \cite{Sahin_2022MVjournal}. 

Our main goal is to obtain a scheme that computes the MVs with a low \ac{CER} while maintaining the \ac{PMEPR} of the transmitted signals as low as possible, where we define the \ac{CER} as the probability of incorrect detection of the $\indexIteration$th \ac{MV} as
$\probability[{\majorityVoteDetectedEle[\indexIteration]\neq\majorityVoteEle[\indexIteration]}]$.
Although the \acp{CS} generated with Theorem~\ref{th:reduced} can address the \ac{PMEPR} challenge by keeping it at most $3$~dB, it is not trivial to use them for the \ac{MV} computation. Therefore, the question that we address is {\em how can we use Theorem~\ref{th:reduced} to develop a reliable \ac{OAC} scheme for computing \ac{MV} without using the \ac{CSI} at the sensors and the \ac{UAV}?}

\section{Proposed Scheme}
\label{sec:scheme}
The proposed scheme modulates of the amplitude of the elements of the \ac{CS} via $\funcfForFinalAmplitude(\seqx)$ as a function of the votes $\voteVectorED[\indexED]\triangleq(\voteVectorEDEle[\indexED][1],\mydots,\voteVectorEDEle[\indexED][\numberOfIterations])$ at the $\indexED$th sensor. To this end,  based on Theorem~\ref{th:reduced}, let us denote the functions used at $\indexED$th sensor as $\funcfForFinalAmplitudeED[\indexED](\seqx)$ and $\funcfForFinalPhaseED[\indexED](\seqx)$, and their parameters as $\{\arbitraryScaleE_\indexED,\scaleEexp[\indexED,1],\mydots,\scaleEexp[\indexED,\numberOfIterations]\}$ and $\{\arbitraryPhaseK_\indexED,\angleexpAll[\indexED,1],\mydots,\angleexpAll[\indexED,\numberOfIterations]\}$, respectively. To synthesize the transmitted sequence $\transmittedSeq[\indexED]$ of length $\lengthGaGb=2^\numberOfIterations$, we use a fixed permutation $\seqPermutationCompShift$ and  map $\voteVectorEDEle[\indexED][\indexIteration]$ to $\scaleEexp[\indexED,\indexIteration]$ as 
\begin{align}
	\scaleEexp[\indexED,\indexIteration]=\begin{cases}
		-\scalingParameters,&\voteVectorEDEle[\indexED][\indexIteration]=-1\\
		0,&\voteVectorEDEle[\indexED][\indexIteration]=0\\
		+\scalingParameters,&\voteVectorEDEle[\indexED][\indexIteration]=+1
	\end{cases},~\forall\indexIteration~,
\label{eq:modulation}
\end{align}
where $\scalingParameters>0$ is a scaling parameter. To ensure that the squared $\ell_2$-norm of the \ac{CS} $\transmittedSeq[\indexED]$ is  $2^\numberOfIterations$, i.e., $\norm{\transmittedSeq[\indexED]}_2^2=2^\numberOfIterations$, we choose $\arbitraryScaleE_\indexED$ as
\begin{align}
\arbitraryScaleE_\indexED =- \frac{1}{2}\sum_{\indexIteration=1}^{\numberOfIterations}\ln{\frac{1+\constante^{2\scaleEexp[\indexIteration]}}{2}}.
\label{eq:normalizationCoef}
\end{align}
To derive \eqref{eq:normalizationCoef}, notice that $\scaleEexp[\indexED,\indexIteration]$ scales $2^{\numberOfIterations-1}$ elements of the \ac{CS} by $\constante^{\scaleEexp[\indexED,\indexIteration]}$ in \eqref{eq:realPartReduced}. Therefore, $\norm{\transmittedSeq[\indexED]}_2^2$ is scaled by ${(1+\constante^{2\scaleEexp[\indexED,\indexIteration]}})/{2}$. By considering $\scaleEexp[\indexED,1],\mydots,\scaleEexp[\indexED,\numberOfIterations]$, the total scaling can be calculated as $\totalPowerScale=\prod_{\indexIteration=1}^{\numberOfIterations}{(1+\constante^{2\scaleEexp[\indexED,\indexIteration]}})/{2}$. Thus, $\constante^{2\arbitraryScaleE_\indexED}=1/\totalPowerScale$ must hold for $\norm{\transmittedSeq[\indexED]}_2^2=2^\numberOfIterations$, which results in \eqref{eq:normalizationCoef}.

With \eqref{eq:modulation} and \eqref{eq:normalizationCoef}, if $\voteVectorEDEle[\indexED][\indexIteration]\neq0$ for $\scalingParameters\rightarrow\infty$, one half of elements (i.e., the ones for $
\monomialAmp[{\permutationMono[{\indexIteration}]}]=0$)  of the \ac{CS} $\transmittedSeq[\indexED]$  are set to $0$  while the other half (i.e., the ones for $
\monomialAmp[{\permutationMono[{\indexIteration}]}]=1$) are scaled by a factor of $\sqrt{2}$ and the sign of $\voteVectorEDEle[\indexED][\indexIteration]$ determines which half is amplified. For $\voteVectorEDEle[\indexED][\indexIteration]=0$, the halves are not scaled.
\begin{example}
	\label{ex:sequences} \rm
\begin{table}[ht]
		\caption{An example of encoded CSs based on votes for $\numberOfIterations=3$.}
		\centering
\resizebox{\textwidth*2/4}{!}{	
	\begin{tabular}{l|c|c|c|c|c|c|c|c}
		$\voteVectorED[\indexED]$ & $\transmittedSeqEle[\indexED,0]$ & $\transmittedSeqEle[\indexED,1]$ & $\transmittedSeqEle[\indexED,2]$ & $\transmittedSeqEle[\indexED,3]$  & $\transmittedSeqEle[\indexED,4]$  & $\transmittedSeqEle[\indexED,5]$ & $\transmittedSeqEle[\indexED,6]$ & $\transmittedSeqEle[\indexED,7]$ \\\hline
	$(0, 0, 0)$	& $1$ & $1$ & $1$ & $-1$ & $1$ & $1$ & $-1$ & $1$	\\	
	$(1, 0, 0)$	& $0$ & $\sqrt{2}$ & $\sqrt{2}$ & $0$ & $0$ & $\sqrt{2}$ & $-\sqrt{2}$ & $0$  \\
	$(1, 1, 0)$	& $0$ & $0$ & $2$ & $0$ & $0$ & $2$ & $0$ & $0$  \\
	$(1, 1, 1)$	& $0$ & $0$ & $0$ & $0$ & $0$ & $2\sqrt{2}$ & $0$ & $0$  \\
	$(1, 1, -1)$& $0$ & $0$ & $2\sqrt{2}$ & $0$ & $0$ & $0$ & $0$ & $0$\\
	$(1, -1, 0)$ & $0$ & $2$ & $0$ & $0$ & $0$ & $0$ & $-2$ & $0$  \\
	$(-1, 0, 0)$ & $\sqrt{2}$ & $0$ & $0$ & $-\sqrt{2}$ & $\sqrt{2}$ & $0$ & $0$ & $\sqrt{2}$  \\	
	\end{tabular}
}
\label{table:example}
\end{table}
Let $\seqPermutationCompShift=(3,2,1)$, $\numberOfPointsForPSK=2$, $\numberOfIterations=3$, $\angleexpAll[\indexIteration]=\arbitraryPhaseK=0$, $\forall\indexIteration$. Hence, the indices of the scaled elements are controlled by $\monomialAmp[{1}]=\monomial[{1}]$, $\monomialAmp[{2}]=(\monomial[{1}] +\monomial[{2}])_2$, and $\monomialAmp[3]=(\monomial[2] +\monomial[3])_2$ when $(\monomial[1],\monomial[2],\monomial[3])$ is listed in lexicographic order, i.e., $(0,0,0),(0,0,1),\mydots,(1,1,1)$.  The encoded \acp{CS} for  several realizations of  $\voteVectorED[\indexED]$ for $\scalingParameters\rightarrow\infty$ are given in \tablename~\ref{table:example}. For $\voteVectorED[\indexED]=(1, 0, 0)$ and $\voteVectorED[\indexED]=(-1, 0, 0)$, four elements determined by $\monomialAmp[3]$ of the uni-modular \ac{CS}  is scaled by $\sqrt{2}$, and the rest is multiplied with $0$. 
Similarly, for $\voteVectorED[\indexED]=(1, 1, 1)$ and $\voteVectorED[\indexED]=(1, 1, -1)$, four elements of the \ac{CS} (i.e., the CS for $\voteVectorED[\indexED]=(1, 1, 0)$) is scaled by $\sqrt{2}$, and the rest is multiplied with $0$. It is worth noting that if all the votes are non-zero, only one of the eight elements of the sequence is non-zero. 
\end{example}

For the proposed scheme, the values for  $\arbitraryPhaseK_\indexED,\angleexpAll[\indexED,1],\mydots,\angleexpAll[\indexED,\numberOfIterations]$ are chosen randomly from the set $\integers_\numberOfPointsForPSK$ for the randomization of $\transmittedSeq[\indexED]$ across the sensors. This choice is also in line with the cases where phase synchronization cannot be maintained in the network.

Based on \eqref{eq:symbolOnSubcarrier}, the received sequence at the \ac{UAV} after signal superposition can be expressed as
\small\begin{align}
	\polySeq[{\receivedSeqP}][\polyVariable] &=\sum_{\forall\seqx\in\integers^\numberOfIterations_2} \underbrace{\left(\sum_{\indexED=1}^{\numberOfEdgeDevices}
		\channelAtSubcarrier[\indexED,\funcEnum(\seqx)]
		\constante^{\funcfForFinalAmplitudeED[\indexED](\seqx)}
		\constante^{\constantj\frac{2\pi}{\numberOfPointsForPSK}\funcfForFinalPhaseED[\indexED](\seqx)}+\noiseAtSubcarrier[\funcEnum(\seqx)]\right)}_{\receivedSeqEle[\funcEnum(\seqx)]}
	\polyVariable^{\funcEnum(\seqx)}~,	
	\label{eq:super}
\end{align}\normalsize
To compute the $\indexIteration$th \ac{MV}, the \ac{UAV} calculates two metrics given by
\small
\begin{align}
\metricPlus[\indexIteration] &\triangleq \sum_{\substack{\forall\seqx\in\integers^\numberOfIterations_2\\\monomialAmp[{\permutationMono[{\indexIteration}]}]=1}}|\receivedSeqEle[\funcEnum(\seqx)]|^2=\sum_{\substack{\forall\seqx\in\integers^\numberOfIterations_2\\\monomialAmp[{\permutationMono[{\indexIteration}]}]=1}} \left|\sum_{\indexED=1}^{\numberOfEdgeDevices}	\channelAtSubcarrier[\indexED,\funcEnum(\seqx)] 	\constante^{\funcfForFinalAmplitudeED[\indexED](\seqx)}	\constante^{\constantj\frac{2\pi}{\numberOfPointsForPSK}\funcfForFinalPhaseED[\indexED](\seqx)}+\noiseAtSubcarrier[\funcEnum(\seqx)]\right|^2\nonumber,
\end{align}\normalsize
and\small
\begin{align}
	\metricMinus[\indexIteration] &\triangleq \sum_{\substack{\forall\seqx\in\integers^\numberOfIterations_2\\\monomialAmp[{\permutationMono[{\indexIteration}]}]=0}}|\receivedSeqEle[\funcEnum(\seqx)]|^2=\sum_{\substack{\forall\seqx\in\integers^\numberOfIterations_2\\\monomialAmp[{\permutationMono[{\indexIteration}]}]=0}} \left|\sum_{\indexED=1}^{\numberOfEdgeDevices}	\channelAtSubcarrier[\indexED,\funcEnum(\seqx)] 	\constante^{\funcfForFinalAmplitudeED[\indexED](\seqx)}	\constante^{\constantj\frac{2\pi}{\numberOfPointsForPSK}\funcfForFinalPhaseED[\indexED](\seqx)}+\noiseAtSubcarrier[\funcEnum(\seqx)]\right|^2\nonumber.
\end{align}\normalsize
It then detects the $\indexIteration$th \ac{MV} by comparing the values of $\metricPlus[\indexIteration]$ and $\metricMinus[\indexIteration]$ as
\begin{align}
	\majorityVoteDetectedEle[\indexIteration]=\signNormal[{\metricPlus[\indexIteration]-\metricMinus[\indexIteration]}],\forall\indexIteration~.
	\label{eq:decision}
\end{align}
We discuss how \eqref{eq:decision} allow the UAV to detect the correct \ac{MV} in the average sense in the following subsection.

\subsection{How Does It Work without CSI?}
Let $\numberOfEDsPlus[\indexIteration]$, $\numberOfEDsMinus[\indexIteration]$, and $\numberOfEDsZero[\indexIteration]$ be the number of sensors with positive, negative, and zero votes for $\indexIteration$th \ac{MV} computation, respectively. 
\begin{lemma}\rm
$\expectationOperator[{\metricPlus[\indexIteration]}][]$ and $\expectationOperator[{\metricMinus[\indexIteration]}][]$ can be calculated as
\begin{align}
	\expectationOperator[{\metricPlus[\indexIteration]}][]
&=\frac{2^\numberOfIterations\constante^{2\scalingParameters}\numberOfEDsPlus[\indexIteration]}{1+\constante^{2\scalingParameters}}+\frac{2^\numberOfIterations\constante^{-2\scalingParameters}\numberOfEDsMinus[\indexIteration]}{1+\constante^{-2\scalingParameters}}+{2^{\numberOfIterations-1}}(\numberOfEDsZero[\indexIteration]+\noiseVariance)	\nonumber,
\\
\expectationOperator[{\metricMinus[\indexIteration]}][]
&=\frac{2^\numberOfIterations\numberOfEDsPlus[\indexIteration]}{1+\constante^{2\scalingParameters}}+\frac{2^\numberOfIterations\numberOfEDsMinus[\indexIteration]}{1+\constante^{-2\scalingParameters}}+{2^{\numberOfIterations-1}}(\numberOfEDsZero[\indexIteration]+\noiseVariance)	\nonumber,
\end{align}
respectively, where the expectation is over the distribution of channel and noise.
\label{lemma:exp}
\end{lemma}
To prove Lemma~\ref{lemma:exp}, we first need the following proposition:
\begin{proposition}\rm 
	The following identities hold:
	\begin{align}
		\sum_{\substack{\forall\seqx\in\integers^\numberOfIterations_2\\\monomialAmp[{\permutationMono[{\indexIteration}]}]=1}}
		\constante^{2\funcfForFinalAmplitudeED[\indexED](\seqx)}=\constante^{2\scaleEexp[\indexED,\indexIteration]}\sum_{\substack{\forall\seqx\in\integers^\numberOfIterations_2\\\monomialAmp[{\permutationMono[{\indexIteration}]}]=0}}
		\constante^{2\funcfForFinalAmplitudeED[\indexED](\seqx)}=\frac{\constante^{2\scaleEexp[\indexED,\indexIteration]}}{1+\constante^{2\scaleEexp[\indexED,\indexIteration]}}2^\numberOfIterations~.\nonumber
	\end{align}
	\label{pro:sum}
\end{proposition}
\begin{proof}[Proof of Proposition~\ref{pro:sum}]
	The first identity is because $\constante^{2\scaleEexp[\indexED,\indexIteration]\monomialAmp[{\permutationMono[{\indexIteration}]}]}=1$ for $\monomialAmp[{\permutationMono[{\indexIteration}]}]=0$. 
	With \eqref{eq:normalizationCoef}, $\norm{\transmittedSeq[\indexED]}_2^2=2^\numberOfIterations$ holds. Hence,
	\begin{align}
		\norm{\transmittedSeq[\indexED]}_2^2=&\sum_{\forall\seqx\in\integers^\numberOfIterations_2}
		\constante^{2\funcfForFinalAmplitudeED[\indexED](\seqx)}=
		\sum_{\substack{\forall\seqx\in\integers^\numberOfIterations_2\\\monomialAmp[{\permutationMono[{\indexIteration}]}]=0}} 
		\constante^{2\funcfForFinalAmplitudeED[\indexED](\seqx)}+\sum_{\substack{\forall\seqx\in\integers^\numberOfIterations_2\\\monomialAmp[{\permutationMono[{\indexIteration}]}]=1}} 
		\constante^{2\funcfForFinalAmplitudeED[\indexED](\seqx)}\nonumber \\=&\sum_{\substack{\forall\seqx\in\integers^\numberOfIterations_2\\\monomialAmp[{\permutationMono[{\indexIteration}]}]=0}}
		\constante^{2\funcfForFinalAmplitudeED[\indexED](\seqx)}+\constante^{2\scaleEexp[\indexED,\indexIteration]}\sum_{\substack{\forall\seqx\in\integers^\numberOfIterations_2\\\monomialAmp[{\permutationMono[{\indexIteration}]}]=0}}\constante^{2\funcfForFinalAmplitudeED[\indexED](\seqx)} = 2^\numberOfIterations~.\nonumber
	\end{align}
\end{proof}
\begin{proof}[Proof of Lemma~\ref{lemma:exp}]
	By using Proposition~\ref{pro:sum}, we can calculate $\expectationOperator[{\metricPlus[\indexIteration]}][]$ and $\expectationOperator[{\metricMinus[\indexIteration]}][]$ as 
	\small\begin{align}
		&\expectationOperator[{\metricPlus[\indexIteration]}][]
		=\sum_{\substack{\forall\seqx\in\integers^\numberOfIterations_2\\\monomialAmp[{\permutationMono[{\indexIteration}]}]=1}} \expectationOperator[{\left|\sum_{\indexED=1}^{\numberOfEdgeDevices}	\channelAtSubcarrier[\indexED,\funcEnum(\seqx)] 	\constante^{\funcfForFinalAmplitudeED[\indexED](\seqx)}	\constante^{\constantj\frac{2\pi}{\numberOfPointsForPSK}\funcfForFinalPhaseED[\indexED](\seqx)}+\noiseAtSubcarrier[\funcEnum(\seqx)]\right|^2}][]
		\nonumber\\
		&=\sum_{\indexED=1}^{\numberOfEdgeDevices} \sum_{\substack{\forall\seqx\in\integers^\numberOfIterations_2\\\monomialAmp[{\permutationMono[{\indexIteration}]}]=1}}
		\constante^{2\funcfForFinalAmplitudeED[\indexED](\varMonomial)}+2^{\numberOfIterations-1}\noiseVariance=2^{\numberOfIterations}\left(\sum_{\indexED=1}^{\numberOfEdgeDevices}
		\frac{\constante^{2\scaleEexp[\indexED,\indexIteration]}}{1+\constante^{2\scaleEexp[\indexED,\indexIteration]}}+\frac{\noiseVariance}{2}\right)\nonumber\\
		&=2^\numberOfIterations\left(\frac{\constante^{2\scalingParameters}}{1+\constante^{2\scalingParameters}}\numberOfEDsPlus[\indexIteration]+\frac{1}{2}\numberOfEDsZero[\indexIteration]+\frac{\constante^{-2\scalingParameters}}{1+\constante^{-2\scalingParameters}}\numberOfEDsMinus[\indexIteration]+\frac{1}{2}\noiseVariance	\right)\nonumber~,
	\end{align}\normalsize
	and
	\small\begin{align}
		&\expectationOperator[{\metricMinus[\indexIteration]}][]
		=\sum_{\substack{\forall\seqx\in\integers^\numberOfIterations_2\\\monomialAmp[{\permutationMono[{\indexIteration}]}]=0}} \expectationOperator[{\left|\sum_{\indexED=1}^{\numberOfEdgeDevices}	\channelAtSubcarrier[\indexED,\funcEnum(\seqx)] 	\constante^{\funcfForFinalAmplitudeED[\indexED](\seqx)}	\constante^{\constantj\frac{2\pi}{\numberOfPointsForPSK}\funcfForFinalPhaseED[\indexED](\seqx)}+\noiseAtSubcarrier[\funcEnum(\seqx)]\right|^2}][]
		\nonumber\\
		&=\sum_{\indexED=1}^{\numberOfEdgeDevices} \sum_{\substack{\forall\seqx\in\integers^\numberOfIterations_2\\\monomialAmp[{\permutationMono[{\indexIteration}]}]=0}}
		\constante^{2\funcfForFinalAmplitudeED[\indexED](\varMonomial)}+2^{\numberOfIterations-1}\noiseVariance=2^{\numberOfIterations}\left(\sum_{\indexED=1}^{\numberOfEdgeDevices}
		\frac{1}{1+\constante^{2\scaleEexp[\indexED,\indexIteration]}}+\frac{\noiseVariance}{2}\right)\nonumber\\
		&=2^\numberOfIterations\left(\frac{1}{1+\constante^{2\scalingParameters}}\numberOfEDsPlus[\indexIteration]+\frac{1}{2}\numberOfEDsZero[\indexIteration]+\frac{1}{1+\constante^{-2\scalingParameters}}\numberOfEDsMinus[\indexIteration]+\frac{1}{2}\noiseVariance	\right)\nonumber~.\end{align}
\end{proof}\normalsize

Without any concern about the norm of  $\transmittedSeq[\indexED]$ with \eqref{eq:normalizationCoef}, we can  choose an arbitrarily large  $\scalingParameters$, leading to following result:
\begin{corollary}
Following identities hold:
\begin{align}
\lim_{\scalingParameters\rightarrow\infty}\expectationOperator[{\metricPlus[\indexIteration]}][]&=2^\numberOfIterations\numberOfEDsPlus[\indexIteration]+2^{\numberOfIterations-1}\numberOfEDsZero[\indexIteration]+2^{\numberOfIterations-1}\noiseVariance~,\nonumber\\
\lim_{\scalingParameters\rightarrow\infty}\expectationOperator[{\metricMinus[\indexIteration]}][]&=2^\numberOfIterations\numberOfEDsMinus[\indexIteration]+2^{\numberOfIterations-1}\numberOfEDsZero[\indexIteration]+2^{\numberOfIterations-1}\noiseVariance~. \nonumber
\end{align}
\label{corol:exp}
\end{corollary} 

Since $\lim_{\scalingParameters\rightarrow\infty}\expectationOperator[{{\metricPlus[\indexIteration]-\metricMinus[\indexIteration]}}][]=2^{\numberOfIterations-1}(\numberOfEDsPlus[\indexIteration]-\numberOfEDsMinus[\indexIteration])$ holds, we can infer that the detector in \eqref{eq:decision} detects the correct \ac{MV} in average for $\scalingParameters\rightarrow\infty$. In other words, although the proposed scheme makes errors instantaneously, in average, the error is centered around the correct MV.

\def\computationRate{\mathcal{R}}
The proposed scheme computes $\numberOfIterations$ \acp{MV} over $2^{\numberOfIterations}$ complex-valued resources. Hence,  the number of functions  computed per channel use (in real dimension)  can be expressed as $	\computationRate={\numberOfIterations}/{2^{\numberOfIterations+1}}
$. Hence, for a larger $\numberOfIterations$, the computation rate reduces, but the \ac{CER} improves  as demonstrated in Section~\ref{sec:numerical}.

\section{Numerical Results}
\label{sec:numerical}
In this section, we first analyze the performance of the scheme for an arbitrary application. Subsequently, we then apply it to the scenario discussed in Section~\ref{sec:system}.
\begin{figure*}
	\centering
	\subfloat[{PMEPR distribution ($\numberOfIterations=8$).}]{\includegraphics[width =\figuresize]{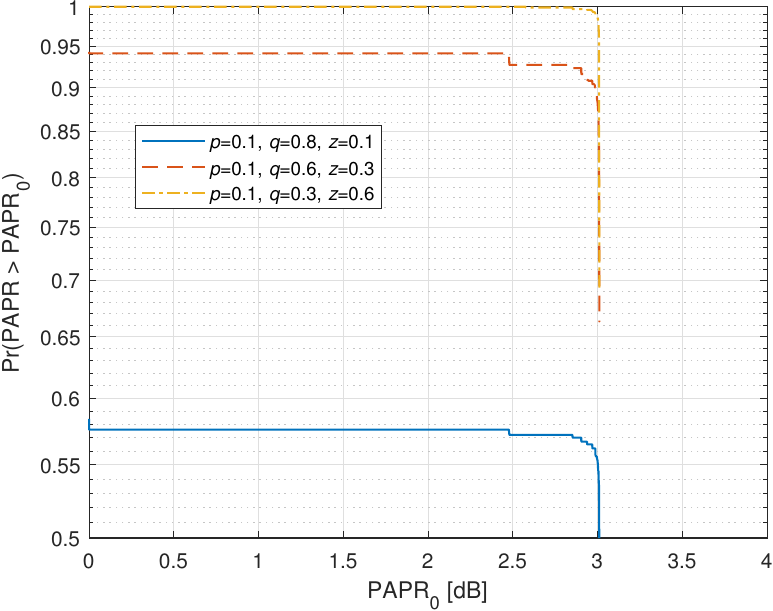}
	\label{subfig:pmepr}}	~~
	\subfloat[{CER in AWGN channel.}]{\includegraphics[width =\figuresize]{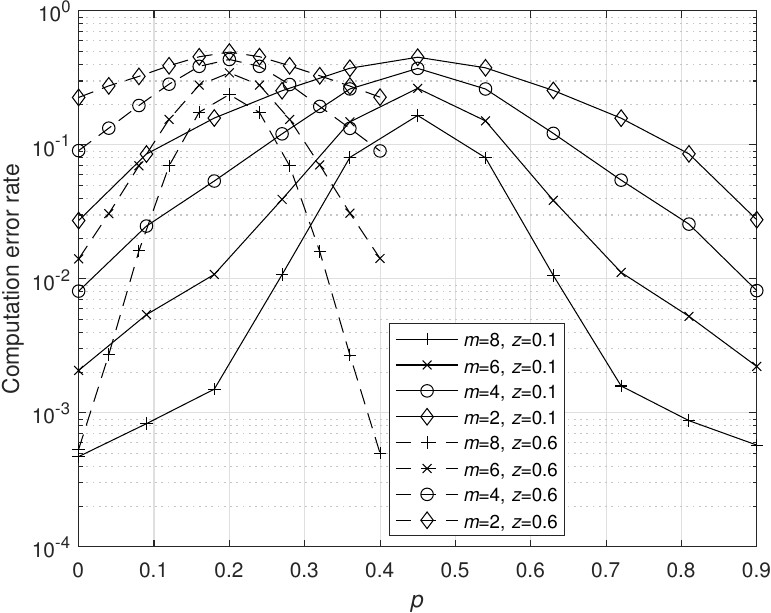}
		\label{subfig:CERAWGN}}	\\
	\subfloat[{CER in flat-fading channel.}]{\includegraphics[width =\figuresize]{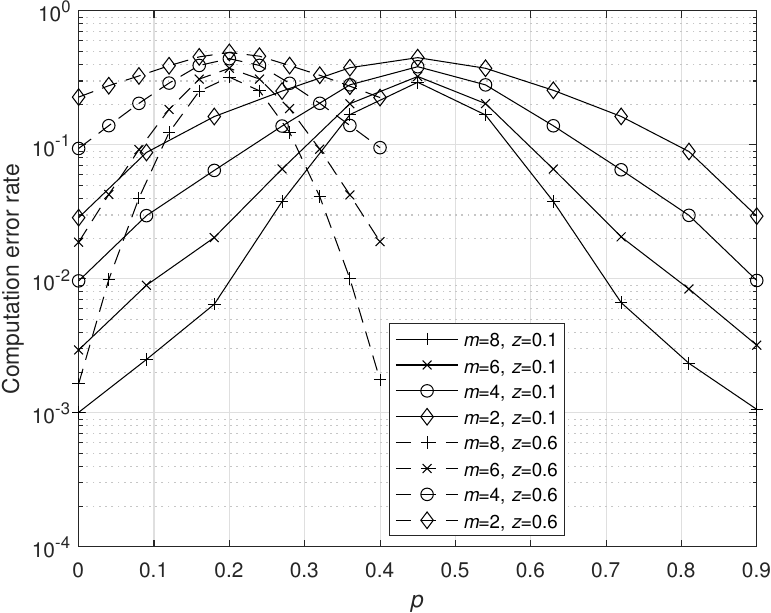}
		\label{subfig:CERflatFading}} 	~~
	\subfloat[{CER in frequency-selective channel.}]{\includegraphics[width =\figuresize]{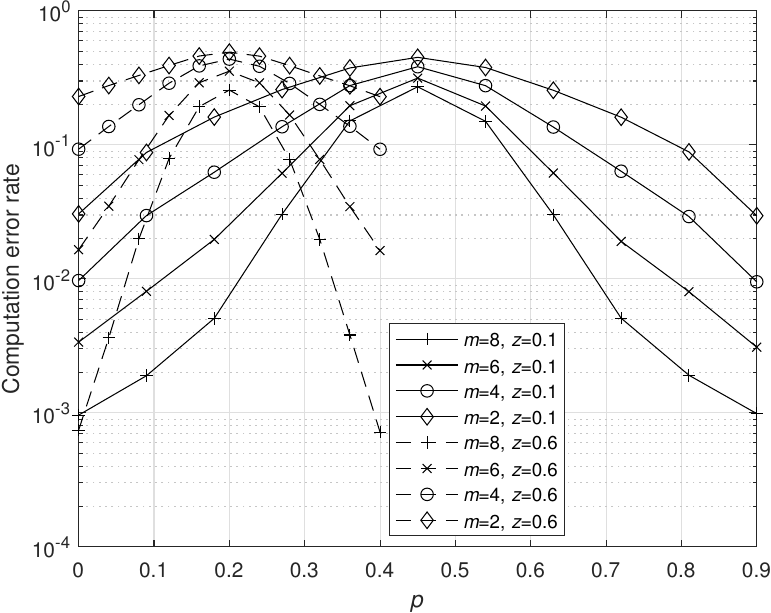}
	\label{subfig:CERselectiveFading}}
	\caption{CER for in AWGN, flat-fading, and frequency-selective channels ($\numberOfEdgeDevices=50$ sensors) and PMEPR distribution.}
	\label{fig:pmeprcer}
\end{figure*}

\begin{figure}
	\centering
	\subfloat[{UAV's trajectory in time.}]{\includegraphics[width =3in]{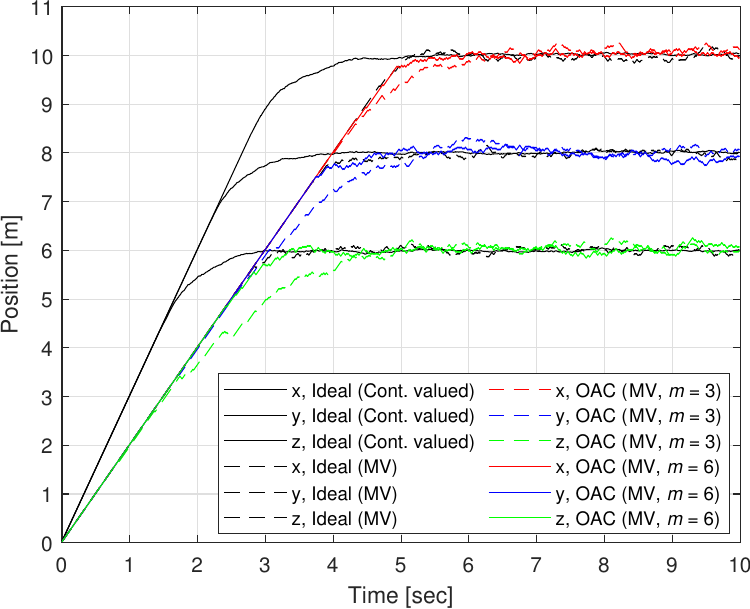}
		\label{subfig:wayPointSingleTime}}	\vspace{-4mm}\\
	\subfloat[{UAV's trajectory in space.}]{\includegraphics[width =3.5in]{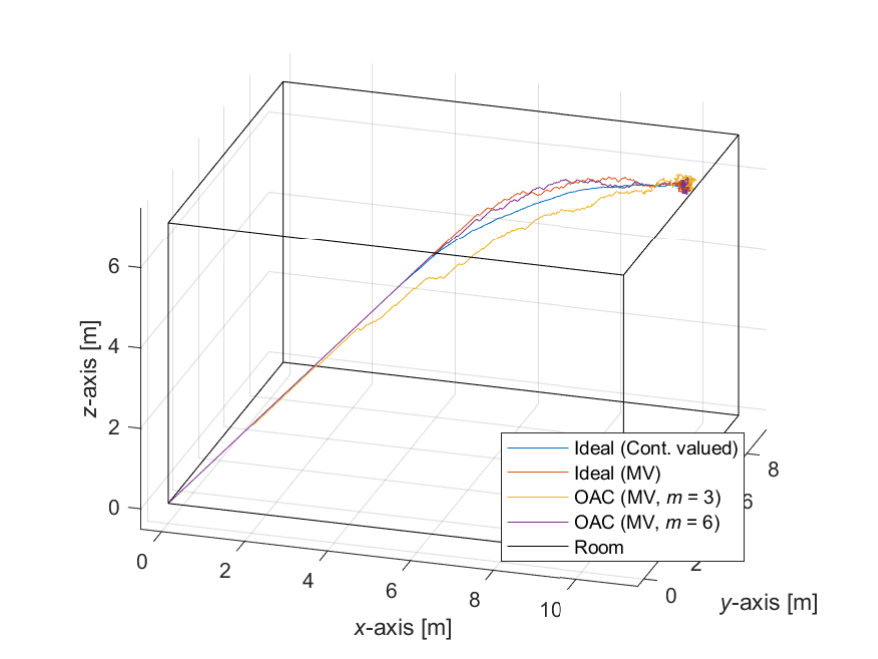}
		\label{subfig:wayPointSingleSpace}}
	\caption{UAV's trajectory with a single point of interest. The initial position is $(0,0,0)$ and the target position is $(10,8,6)$.}
	\label{fig:waypointSingle}
\end{figure}
\begin{figure}
	\centering
	\subfloat[UAV's trajectory in time.]{\includegraphics[width =3in]{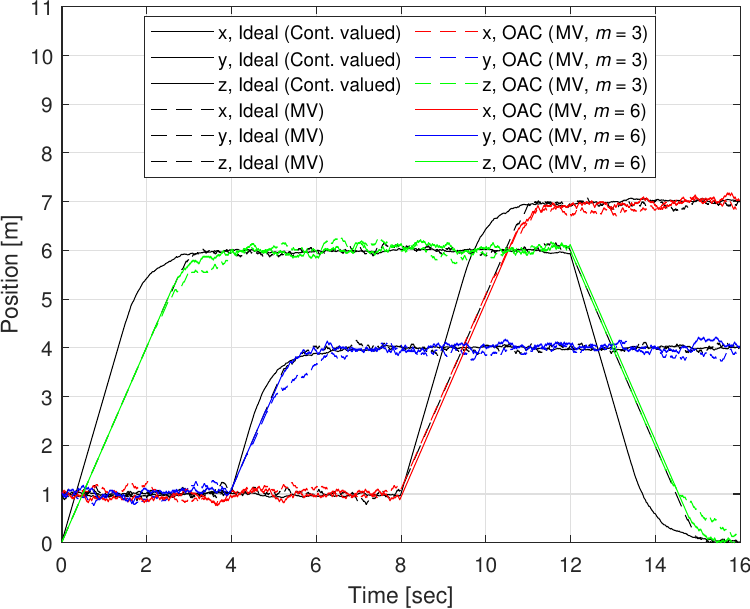}
		\label{subfig:wayPointTime}}\vspace{-4mm}\\	
	\subfloat[{UAV's trajectory in space.}]{\includegraphics[width =3.5in]{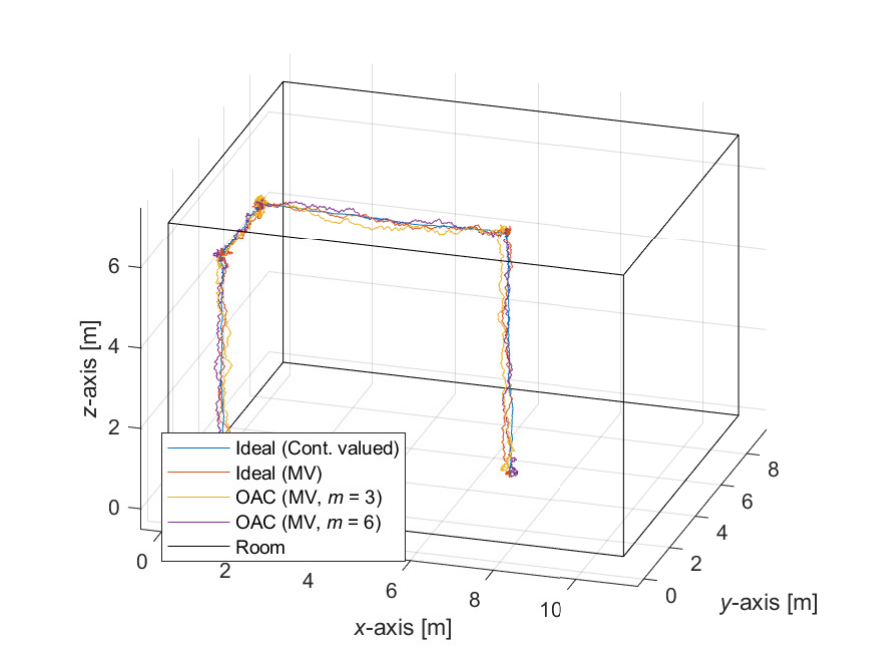}
		\label{subfig:waypointSpace}}		
	\caption{UAV's trajectory with multiple points of interest, i.e., $(1,1,6)$, $(1,4,6)$, $(7,4,6)$, and $(7,4,0)$. The initial position  is $(1,1,0)$.}
	\label{fig:waypoint}
\end{figure}
In \figurename~\ref{fig:pmeprcer}, we provide the \ac{PMEPR} distribution and the \ac{CER} for a given $\{\probabilityPlusDecision,\probabilityMinusDecision,\probabilityNullDecision\}$, where $\probabilityPlusDecision$, $\probabilityMinusDecision$, and $\probabilityNullDecision$ denote the probabilities  
$ \probability[{\voteVectorEDEle[\indexED][\indexIteration]>0}]$, $ \probability[{\voteVectorEDEle[\indexED][\indexIteration]<0}]$, and $ \probability[{\voteVectorEDEle[\indexED][\indexIteration]= 0}]$, respectively,  for $\numberOfEdgeDevices=50$ sensors. In \figurename~\ref{fig:pmeprcer}\subref{subfig:pmepr}, we provide \ac{PMEPR} distribution for $\numberOfIterations=8$,  $\probabilityPlusDecision=0.1$ and $\probabilityNullDecision\in\{0.1,0.3,0.6\}$. We show that the \ac{PMEPR} of the proposed scheme is always less than or equal to $3$~dB due to the properties of the \acp{CS}. If there are no absentee votes, the maximum \ac{PMEPR} of the proposed scheme is $0$~dB since a single subcarrier is used for the transmission. Hence, for a larger absentee vote probability, the probability of observing $0$~dB PMEPR increases. The combination of sequences that lead to $0$~dB  and $3$~dB PMEPR values result in the jumps in the PMEPR distribution given in \figurename~\ref{fig:pmeprcer}\subref{subfig:pmepr}. 
In \figurename~\ref{fig:pmeprcer}\subref{subfig:CERAWGN}-\subref{subfig:CERselectiveFading}, we analyze \ac{CER} for  $\probabilityNullDecision\in\{0.1,0.6\}$ and $\numberOfIterations=\{2,4,6,8\}$ by sweeping $\probabilityPlusDecision$ in \ac{AWGN}  (i.e., $\channelAtSubcarrier[\indexED,\varMonomial]=1$), flat-fading   (i.e., $\channelAtSubcarrier[\indexED,\varMonomial]=\channelAtSubcarrier[\indexED,\varMonomial']\sim\complexGaussian[0][1]$), and frequency-selective   (i.e., $\channelAtSubcarrier[\indexED,\varMonomial]\sim\complexGaussian[0][1]$) channels, respectively. We observe that the proposed scheme achieves a better \ac{CER} for increasing $\numberOfIterations$ at the expense of more resource consumption in all channel conditions. The \ac{CER} improves for a small or a large $\probabilityPlusDecision$ since more sensors vote for $-1$ or $+1$, respectively. The performance in frequency selective channel is slightly better than the flat-fading channel because of the diversity~gain.

 In \figurename~\ref{fig:waypointSingle} and  \figurename~\ref{fig:waypoint},  we consider the distributed UAV guidance scenario discussed in Section~\ref{sec:system} for $\numberOfEdgeDevices=50$ sensors.  We assume $\Trefresh=10$~ms, $\updateRate=2$, $\maximumVelocity = 3$~m/s, $\varianceSensor=2$, and $\SNR=10$~dB. We  provide the trajectory of the \ac{UAV} in time and space for the aforementioned  waypoint flight control scenario.  We consider two cases. In the first case, there is only one point of interest $(\locationTargetEle[1],\locationTargetEle[2],\locationTargetEle[3])=(10,8,6)$ and the initial position of the \ac{UAV} is $(0,0,0)$. In the second case, the points of interest are $(1,1,6)$, $(1,4,6)$, $(7,4,6)$, and $(7,4,0)$, where the initial position of the \ac{UAV} is $(1,1,0)$. We compare the  proposed scheme for $\numberOfIterations\in\{3,6\}$ in a practical communication channel with both continuous and \ac{MV}-based feedback in an ideal communication channel (i.e., no error due to the communications).
   As can be seen from \figurename~\ref{fig:waypointSingle}\subref{subfig:wayPointSingleTime}, for the continuous-valued feedback, the \ac{UAV} reaches its position faster than any \ac{MV}-based approach. This is because the velocity increment is limited with the step size for \ac{MV}-based feedback in our setup. Hence, as can be seen from \figurename~\ref{fig:waypointSingle}\subref{subfig:wayPointSingleSpace}  the \ac{UAV}'s trajectory in space is slightly bent. Since the proposed scheme is also based on the \ac{MV} computation, its characteristics are similar to the one with MV computation in an ideal channel. Since the \ac{CER} with $\numberOfIterations=6$ is lower than the one with  $\numberOfIterations=3$, the proposed scheme for $\numberOfIterations=6$ performs better and its characteristics are similar to the ideal \ac{MV}-based feedback. The position of the \ac{UAV} in time and space for multiple points of interest is given in \figurename~\ref{fig:waypoint}\subref{subfig:wayPointTime} and \figurename~\ref{fig:waypoint}\subref{subfig:waypointSpace}, respectively. The proposed scheme for $\numberOfIterations=6$ performs similarly to the one with the  \acp{MV} in ideal communications and increasing $\numberOfIterations$ leads to a more stable trajectory.

\section{Concluding Remarks}
In this study, we modulate the amplitude of the \ac{CS} based on Theorem~\ref{th:reduced} to develop a new non-coherent \ac{OAC} scheme for \ac{MV} computation. We show that the proposed scheme reduces the \ac{CER} via bandwidth expansion in both flat-fading and frequency-selective fading channel conditions while maintaining the \ac{PMEPR} of the transmited signals to be less than or equal to $3$~dB. We demonstrate the applicability of the proposed OAC scheme to an indoor  flight control scenario and show that the proposed scheme performs similar to the case where \ac{MV} without OAC with larger length of sequences. In the future work, we provide the theoretical computation error and convergence analyses for distributed UAV guidance.

\bibliographystyle{IEEEtran}
\bibliography{references}

\end{document}

%% file: variables.tex
\def\figuresize{3.0in}
\def\complexNumbers{\mathbb{C}}

\def\integersPositive{\mathbb{Z}^{+}}

\def\functionSpace[#1]{\mathcal{F}(#1)}
\def\realNumbers{\mathbb{R}}
\def\integers{\mathbb{Z}}
\def\constante{{\rm e}}
\def\constantj{{\rm j}}
\def\expectationOperator[#1][#2]{\mathbb{E}\left[#1\right]}
\def\uniformDistribution[#1][#2]{{\mathcal{U}_{[#1,#2]}}}
\def\traceOperator[#1]{{\mathrm{tr}}\{#1\}}
\def\identityMatrix[#1]{\textbf{\mathrm{I}}_{#1}}
\def\zeroVector[#1]{\textbf{\mathrm{0}}_{#1}}
\def\exponentialIntegral[#1]{\mathrm{Ei}(#1)}
\def\functionArbitrary[#1]{f_{#1}}
\def\functionArbitraryEstimate[#1]{\hat{f}_{#1}}

\def\clamp[#1][#2]{\text{clamp}_{#2}\left(#1\right)}
\def\signNormal[#1]{\text{sign}\left(#1\right)}
\def\diagOperation[#1]{\text{diag}\left\{#1\right\}}

\def\probability[#1]{\mathrm{Pr}\left({#1}\right)}
\def\complexGaussian[#1][#2]{\mathcal{CN}({#1,#2})}
\def\gaussian[#1][#2]{\mathcal{N}({#1,#2})}

\def\normalPDF[#1]{\phi\left(#1\right)}
\def\normalCDF[#1]{\Phi\left(#1\right)}

\def\CDF[#1][#2][#3]{F_{#1}^{#2}\left({#3}\right)}

\def\PDF[#1][#2][#3]{f_{#1}^{#2}\left({#3}\right)}

\def\numberOfEdgeDevices{K}
\def\indexED{k}

\def\voteVectorEDEle[#1][#2]{v_{#1,#2}}
\def\voteVectorED[#1]{{\textit{\textbf{v}}}_{#1}}

\def\majorityVoteEle[#1]{\omega_{#1}}
\def\majorityVoteDetectedEle[#1]{\hat{\omega}_{#1}}

\def\varMonomial{{i_{\rm d}}}
\def\monomial[#1]{x_{#1}}
\def\monomialAmp[#1]{y_{#1}}

\def\lengthGaGb{L}

\def\scalingParameters{\alpha}
\def\scaleAexp[#1]{a_{#1}}
\def\scaleBexp[#1]{b_{#1}}
\def\scaleEexp[#1]{a_{#1}}
\def\angleexpAll[#1]{c_{#1}}
\def\angleScaleAexp[#1]{\dot{c}_{#1}}
\def\angleScaleBexp[#1]{\ddot{c}_{#1}}
\def\arbitraryScaleE{a'}

\def\arbitraryPhaseK{c'}

\def\separationGolay[#1]{d_{#1}}

\def\angleexpAllBitD[#1]{\hat{k}_{#1}}

\def\permutationMonoD[#1]{{\hat{\pi}_{#1}}}

\def\eleGa[#1]{{a}_{#1}}
\def\eleGb[#1]{{b}_{#1}}
\def\eleGc[#1]{{c}_{#1}}
\def\eleGt[#1]{{t}_{#1}}

\def\angleexpAllBitD[#1]{\hat{k}_{#1}}

\def\permutationMonoD[#1]{{\hat{\pi}_{#1}}}

\def\eleGa[#1]{{a}_{#1}}
\def\eleGb[#1]{{b}_{#1}}
\def\eleGc[#1]{{c}_{#1}}
\def\eleGt[#1]{{t}_{#1}}

\def\apac[#1][#2]{\rho_{#1}(#2)}
\def\apacPositive[#1][#2]{\rho^{+}_{#1}(#2)}
\def\binaryAsignment[#1][#2]{b_{#1}^{(#2)}}
\def\pmfBinomial[#1][#2]{{p}_{\rm b}(#1;#2)}

\def\eleSeqf[#1]{{f}_{#1}}
\def\eleSeqg[#1]{{g}_{#1}}
\def\eleSeqcf[#1]{{c}_{f,#1}}
\def\eleSeqcg[#1]{{c}_{g,#1}}

\def\scaleA[#1]{\alpha_{#1}}
\def\scaleB[#1]{\beta_{#1}}
\def\angleGolay[#1]{\omega_{#1}}
\def\angleScaleA[#1][#2]{\dot{\omega}_{#1}^{#2}}
\def\angleScaleB[#1][#2]{\ddot{\omega}_{#1}^{#2}}
\def\permutationShift[#1]{{\psi_{#1}}}

\def\permutationMono[#1]{{\pi_{#1}}}
\def\permutationMonoPre[#1]{{\pi'_{#1}}}

\def\seqPermutationCompShift{\bm{\pi}}

\def\vecArrangement[#1]{\textbf{b}_{#1}}

\def\seqGa{\textit{\textbf{a}}}
\def\seqGb{\textit{\textbf{b}}}
\def\seqGaIt[#1]{\textit{\textbf{a}}^{(#1)}}
\def\seqGbIt[#1]{\textit{\textbf{b}}^{(#1)}}

\def\seqGf[#1]{\textit{\textbf{f}}_{#1}}
\def\seqGg[#1]{\textit{\textbf{g}}_{#1}}
\def\seqGfdot[#1]{\bar{\textit{\textbf{f}}}_{#1}}
\def\seqGgdot[#1]{\bar{\textit{\textbf{g}}}_{#1}}

\def\OFDMinTime[#1][#2]{s_{#1}(#2)}
\def\seqx{\textit{\textbf{x}}}
\def\seqGaP{A}

\def\seqGaItP[#1]{{A}^{(#1)}}
\def\seqGbItP[#1]{{B}^{(#1)}}

\def\seqSubP[#1]{{H}_{#1}}
\def\seqGfP[#1]{{{F}}_{#1}}
\def\seqGgP[#1]{{{G}}_{#1}}

\def\seqSub[#1]{\textit{\textbf{h}}_{#1}}
\def\lagForCorrelation{k}

\def\functionSpace{\mathcal{F}}
\def\functionf[#1]{p^{(#1)}}
\def\functiong[#1]{q^{(#1)}}

\def\functionfdot[#1]{{p}_{\indexIterationANF}^{(#1)}} 
\def\functiongdot[#1]{{q}_{\indexIterationANF}^{(#1)}} 

\def\funcfForFinalPhase{f_{\rm i}}

\def\funcfForFinalAmplitude{f_{\rm r}}

\def\funcfForFinalPhaseED[#1]{f_{{\rm i},#1}}
\def\funcfForFinalPhaseDecED[#1]{\check{f}_{{\rm i},#1}}
\def\funcfForFinalAmplitudeED[#1]{f_{{\rm r},#1}}
\def\funcfForFinalAmplitudeDecED[#1]{\check{f}_{{\rm r}, #1}}

\def\funcfForANF{f}

\def\funcEnum{{i}}

\def\varMonomial{{i}}

\def\funcGfForANF[#1]{f_{#1}}
\def\funcGgForANF[#1]{g_{#1}}
\def\polySeq[#1][#2]{{#1}(#2)}

\def\setMonomials[#1]{{\mathcal{M}_{#1}}}

\def\phaseOffsetl[#1]{\Delta_{#1}}
\def\phaseOffsetll[#1]{\Delta_{#1}}
\def\amplitudeOffset[#1]{\epsilon'_{#1}}
\def\amplitudeOffsetl[#1]{\epsilon_{#1}}

\def\lengthOfSequence{L}

\def\totalPowerScale{\gamma}

\def\indexEleOfSeq{i}
\def\indexIteration{n}

\def\indexFirstOrderMonomial{j}

\def\orderMonomial[#1]{k_{#1}}
\def\coeffientsANF[#1]{c_{#1}}
\def\polyVariable{z}

\def\numberOfIterations{m}
\def\numberOfPointsForPSK{H}

\def\modulationSymbolF[#1]{m_{#1}}
\def\cardinalitySetOfOperators[#1]{{H}_{#1}}

\def\transmittedSeqP{T}
\def\transmittedSeq[#1]{\textit{\textbf{t}}_{#1}}
\def\transmittedSeqEle[#1]{{t}_{#1}}

\def\receivedSeqP{R}
\def\receivedSeq[#1]{\textit{\textbf{r}}_{#1}}
\def\receivedSeqEle[#1]{{r}_{#1}}

\def\SNR{{\tt{SNR}}}
\def\channelAtSubcarrier[#1]{h_{#1}}
\def\noiseAtSubcarrier[#1]{w_{#1}}
\def\transmitPower[#1]{P_{#1}}
\def\noiseVariance{\sigma_{\rm noise}^2}

\def\numberOfEDsPlus[#1]{K^{+}_{#1}}
\def\numberOfEDsMinus[#1]{K^{-}_{#1}}
\def\numberOfEDsZero[#1]{K^{0}_{#1}}
\def\metricPlus[#1]{E^{+}_{#1}}
\def\metricMinus[#1]{E^{-}_{#1}}

\def\probabilityError[#1]{{\tt{CER}}(#1)}
\def\rate[#1]{\lambda_{#1}}

\def\Trefresh{T_\mathrm{update}}
\def\indexRound{\ell}

\def\updateRate{\mu}
\def\velocityVectorEle[#1][#2]{{u}^{(#1)}_{#2}}
\def\feedbackEle[#1][#2]{{g}^{(#1)}_{#2}}
\def\locationInitialEle[#1]{{c}_{\mathrm{i},#1}}
\def\locationTargetEle[#1]{{c}_{\mathrm{t},#1}}
\def\locationEle[#1][#2]{{{p}}^{({#1})}_{#2}}
\def\locationEstimateEle[#1][#2]{{{\tilde{p}}}^{({#1})}_{#2}}
\def\locationEstimationErrorEle[#1][#2]{{n}^{(#1)}_{#2}}
\def\varianceSensor{\sigma_{\rm sensor}^2}
\def\indexCoordinate{l}
\def\maximumVelocity{u_\text{\rm limit}}
\def\clamp[#1][#2]{\text{clamp}_{#2}\left(#1\right)}

\def\probabilityPlusDecision{p}
\def\probabilityNullDecision{z}
\def\probabilityMinusDecision{q}

%% file: otherDefinitions.tex
\let\norm\undefined 
\DeclarePairedDelimiter\norm{\lVert}{\rVert}
\newcommand\mydots{\hbox to 1em{.\hss.\hss.}}

\newtheorem{theorem}{Theorem}

\newtheorem{lemma}{Lemma}
\newtheorem{proposition}{Proposition}
\newtheorem{corollary}{Corollary}
 
\newtheorem{example}{\color{black} Example} 

\makeatletter
\newif\ifAC@uppercase@first%
\def\Aclp#1{\AC@uppercase@firsttrue\aclp{#1}\AC@uppercase@firstfalse}%
\def\AC@aclp#1{%
	\ifcsname fn@#1@PL\endcsname%
	\ifAC@uppercase@first%
	\expandafter\expandafter\expandafter\MakeUppercase\csname fn@#1@PL\endcsname%
	\else%
	\csname fn@#1@PL\endcsname%
	\fi%
	\else%
	\AC@acl{#1}s%
	\fi%
}%
\def\Acp#1{\AC@uppercase@firsttrue\acp{#1}\AC@uppercase@firstfalse}%
\def\AC@acp#1{%
	\ifcsname fn@#1@PL\endcsname%
	\ifAC@uppercase@first%
	\expandafter\expandafter\expandafter\MakeUppercase\csname fn@#1@PL\endcsname%
	\else%
	\csname fn@#1@PL\endcsname%
	\fi%
	\else%
	\AC@ac{#1}s%
	\fi%
}%
\def\Acfp#1{\AC@uppercase@firsttrue\acfp{#1}\AC@uppercase@firstfalse}%
\def\AC@acfp#1{%
	\ifcsname fn@#1@PL\endcsname%
	\ifAC@uppercase@first%
	\expandafter\expandafter\expandafter\MakeUppercase\csname fn@#1@PL\endcsname%
	\else%
	\csname fn@#1@PL\endcsname%
	\fi%
	\else%
	\AC@acf{#1}s%
	\fi%
}%
\def\Acsp#1{\AC@uppercase@firsttrue\acsp{#1}\AC@uppercase@firstfalse}%
\def\AC@acsp#1{%
	\ifcsname fn@#1@PL\endcsname%
	\ifAC@uppercase@first%
	\expandafter\expandafter\expandafter\MakeUppercase\csname fn@#1@PL\endcsname%
	\else%
	\csname fn@#1@PL\endcsname%
	\fi%
	\else%
	\AC@acs{#1}s%
	\fi%
}%
\edef\AC@uppercase@write{\string\ifAC@uppercase@first\string\expandafter\string\MakeUppercase\string\fi\space}%
\def\AC@acrodef#1[#2]#3{%
	\@bsphack%
	\protected@write\@auxout{}{%
		\string\newacro{#1}[#2]{\AC@uppercase@write #3}%
	}\@esphack%
}%
\def\Acl#1{\AC@uppercase@firsttrue\acl{#1}\AC@uppercase@firstfalse}
\def\Acf#1{\AC@uppercase@firsttrue\acf{#1}\AC@uppercase@firstfalse}
\def\Ac#1{\AC@uppercase@firsttrue\ac{#1}\AC@uppercase@firstfalse}
\def\Acs#1{\AC@uppercase@firsttrue\acs{#1}\AC@uppercase@firstfalse}

%% file: acronyms.tex
\acrodef{WSN}{wireless sensor network}
\acrodef{USRP}{universal software radio peripheral}
\acrodef{SN}{sensor node}
\acrodef{FC}{fusion center}
\acrodef{MAC}{multiple-access channel}
\acrodef{FL}{federated learning}
\acrodef{ED}{edge device}
\acrodef{CS}{compressed sensing}
\acrodef{ES}[BS]{base station}
\acrodef{DCN}{data center network}
\acrodef{RIS}{reconfigurable intelligent surfaces}
\acrodef{IMC}{in-memory computing}
\acrodef{FPGA}{field-programmable gate array}
\acrodef{SDR}{software-defined radio}
\acrodef{PS}{processing system}
\acrodef{SS}{soft synchronization}
\acrodef{IQ}{in-phase/quadrature}
\acrodef{IP}{intellectual property}
\acrodef{DMA}{direct-memory access}
\acrodef{RAM}{random access memory}
\acrodef{CC}{companion computer}
\acrodef{FEE}{function estimation error}
\acrodef{MSK}{minimum-shift keying}
\acrodef{TDMA}{time-domain multiple access}
\acrodef{PLNC}{physical-layer network coding}
\acrodef{UAV}{unmanned aerial vehicle}
\acrodef{LoRa}{Long-Range}
\acrodef{DC}{direct-current}
\acrodef{DAC}{digital-to-analog converter}
\acrodef{ADC}{anlog-to-digital converter}
\acrodef{CS}{complementary sequence}
\acrodef{GCP}{Golay complementary pair}
\acrodef{ANF}{algebraic normal form}
\acrodef{AACF}{aperiodic auto-correlation function}
\acrodef{RM}{Reed-Muller}

\acrodef{PUCCH}{physical uplink control channel}

\acrodef{OBO}{output-power back-off}
\acrodef{ACLR}{adjacent-channel-leakage ratio}

\acrodef{LDPC}{low-density parity check}

\acrodef{PDF}{probability density function}
\acrodef{CDF}{cummulative distribution function}

\acrodef{TBMA}{type-based multiple access}

\acrodef{MSFE}{mean-squared function error}
\acrodef{FEE}{function-estimation error}
\acrodef{CER}{computation error rate}
\acrodef{BCER}{block-computation error rate}
\acrodef{CFO}{carrier frequency offset}
\acrodef{TO}{time offset}
\acrodef{PO}{phase offset}
\acrodef{RSSI}{received signal strength  information}

\acrodef{STLC}{space-time line code}
\acrodef{CCI}{co-channel interference}
\acrodef{CSIT}[CSIT]{\ac{CSI} at the transmitter}
\acrodef{CSIR}[CSIR]{\ac{CSI} at the receiver}
\acrodef{MIMO}{multiple-input-multiple-output}
\acrodef{PC}{phase correction}
\acrodef{ZF}{zero-forcing}
\acrodef{ANOVA}{analysis of variance}

\acrodef{PCA}{principal component analysis}
\acrodef{TIG}{Technical Interest Group}

\acrodef{FSK}{frequency-shift keying}
\acrodef{PPM}{pulse-position modulation}
\acrodef{PAM}{pulse-amplitude modulation}

\acrodef{MRC}{maximum-ratio combining}
\acrodef{HP}{hard-coded participation}
\acrodef{HPA}{hard-coded participation with absentees}
\acrodef{SP}{soft-coded participation}
\acrodef{FSK-MV}{\ac{FSK}-based \ac{MV}}
\acrodef{RF}{radio-frequency}
\acrodef{MF}{matched filter}
\acrodef{PPM}{pulse-position modulation}
\acrodef{CSK}{chirp-shift keying}
\acrodef{PPM-MV}[PPM-MV]{\ac{PPM}-based \ac{MV}}
\acrodef{DFT-s-OFDM}{\ac{DFT}-spread \ac{OFDM}}
\acrodef{SC}{single-carrier}
\acrodef{SGD}{stochastic gradient descent}
\acrodef{signSGD}{sign stochastic gradient descent}

\acrodef{SL}{split learning}
\acrodef{SNR}{signal-to-noise ratio}
\acrodef{RMSE}{root-mean-squared error}
\acrodef{OFDM}{orthogonal frequency division multiplexing}
\acrodef{DFT}{discrete Fourier transform}
\acrodef{PSK}{phase-shift keying}
\acrodef{QAM}{quadrature amplitude modulation}
\acrodef{QPSK}{quadrature phase-shift keying}
\acrodef{PMEPR}{peak-to-mean envelope power ratio}
\acrodef{BER}{bit-error ratio}
\acrodef{SNR}{signal-to-noise ratio}
\acrodef{PSD}{power spectral density}
\acrodef{SE}{spectral efficiency}
\acrodef{CP}{cyclic prefix}
\acrodef{AWGN}{additive white Gaussian noise}
\acrodef{CFR}{channel frequency response}
\acrodef{CIR}{channel impulse response}
\acrodef{MMSE}{minimum mean-squared error}
\acrodef{LMMSE}{linear minimum mean-squared error}
\acrodef{BPSK}{binary phase shift keying}
\acrodef{BPSK}{quadrature phase shift keying}
\acrodef{BLER}{block-error rate}
\acrodef{ML}{maximum likelihood}
\acrodef{PHY}{physical layer}
\acrodef{PA}{power amplifier}
\acrodef{IDFT}{inverse DFT}
\acrodef{DoF}{degrees-of-freedom}
\acrodef{IoT}{Internet-of-Things}
\acrodef{FDE}{frequency-domain equalization}
\acrodef{RF}{radio-frequency}
\acrodef{IM}{index modulation}
\acrodef{MF}{matched filter}
\acrodef{PPM}{pulse-position modulation}

\acrodef{MSE}{mean-squared error}
\acrodef{MRT}{maximum-ratio transmission}
\acrodef{ERC}{equal-ratio combining}
\acrodef{BAA}{broadband analog aggregation}
\acrodef{OBDA}{one-bit broadband digital aggregation}
\acrodef{FEEL}{federated edge learning}
\acrodef{FL}{federated learning}
\acrodef{UL}{uplink}
\acrodef{DL}{downlink}
\acrodef{OAC}{over-the-air computation}
\acrodef{TCI}{truncated-channel inversion}
\acrodef{MV}{majority vote}
\acrodef{CNN}{convolution neural network}
\acrodef{ReLU}{rectified-linear unit}
\acrodef{CSI}{channel state information}
\acrodef{PAPR}{peak-to-average power ratio}
\acrodef{SC}{single-carrier}
\acrodef{iid}[IID]{independent and identically distributed}
\acrodef{RMS}{root-mean-square}
\acrodef{4G}{Fourth Generation}
\acrodef{5G}{Fifth Generation}
\acrodef{NR}{New Radio}
\acrodef{LTE}{Long-Term Evolution}
\acrodef{DFT-s-OFDM}{\ac{DFT}-spread \ac{OFDM}}
\acrodef{OFDMA}{orthogonal frequency division multiple access}
\acrodef{HARQ}{hybrid automatic repeat request}
\acrodef{D2D}{Device-to-Device}
\acrodef{NOMA}{non-orthogonal multiple access}
\acrodef{OMA}{orthogonal multiple access}

\acrodef{IMT}{International Mobile Telecommunications}
\acrodef{ITU}{International Telecommunication Union}